\documentclass[onecolumn,draftcls]{IEEEtran}

\usepackage{mystyle}
\usepackage{booktabs}
\usepackage{dsfont,xcolor}
\usepackage{stmaryrd}

\usepackage{xspace}
\usepackage{bbm}
\newcommand{\Dc}{\mathcal{D}}
\newcommand{\Ec}{\mathcal{E}}

\newcommand{\Nc}{\mathcal{N}}

\newcommand{\Uc}{\mathcal{U}}

\newcommand{\Xc}{\mathcal{X}}

\newcommand{\Xh}{{\hat{X}}}

\DeclareMathOperator\E{E}

\newcommand\ie{i.e.,\xspace}
\def\textiid{i.i.d.\@\xspace}
\newcommand\iid{\ifmmode\text{ i.i.d. } \else \textiid \fi}

\newcommand{\ind}{\mathbbmss{1}}

\title{Denoising of  structured random processes}

\author{Wenda Zhou\thanks{Wenda Zhou is with the Statistics Department, Columbia University, email: wz2335@columbia.edu}, Shirin Jalali\thanks{Shirin Jalali is with the Mathematics and Algorithms Group in  Nokia  Bell Labs,  e-mail: Shirin.Jalali@nokia-bell-labs.com}}
\newcommand{\regquant}{r_b}

\newcommand{\diststru}{d}
\newcommand{\Xbbf}{{\mathbb{\mathbf{X}}}}

\newcommand{\Sbbf}{{\mathbb{\mathbf{S}}}}

\newcommand{\QuantInt}[2]{\llbracket #1 \rrbracket_{#2}}

\newtheorem{example}{Example}

\begin{document}

\maketitle

\begin{abstract}
Denoising a stationary process $(X_i)_{i \in \ZZ}$ corrupted by additive white
Gaussian noise $(Z_i)_{i \in \ZZ}$, \ie recovering $X^n$ from $Y^n=X^n+Z^n$,
is a classic and fundamental problem in information theory and statistical
signal processing. 
denoising algorithms, for general analog sources, theoretically-founded
computationally-efficient methods are yet to be found.
In a Bayesian setup, given the distribution of $X^n$, a minimum mean square
error (MMSE) denoiser computes $\E[X^n|Y^n]$. However, for general sources,
computing $\E[X^n|Y^n]$ is computationally very challenging, if not
infeasible. In this paper, starting from a Bayesian setup, a novel denoising
method, namely, quantized maximum a posteriori (Q-MAP) denoiser, is proposed
and its asymptotic performance is analyzed. Both for memoryless sources, and
for structured first-order Markov sources, it is shown that, asymptotically,
as $\sigma^2$ (noise variance) converges to zero, ${1\over
\sigma^2}\E[(X_i-\Xh^{\rm Q-MAP}_i)^2]$ converges to the information
dimension of the source. For the studied memoryless sources, this limit is
known to be the optimal. A key advantage of the Q-MAP denoiser, unlike an
MMSE denoiser, is that it highlights the key properties of the source
distribution that are to be used in its denoising. This property dramatically
reduces the computational complexity of approximating the solution of the
Q-MAP denoiser. Additionally, it naturally leads to a learning-based
denoiser. Using ImageNet database for training, initial simulation results
exploring the performance of such a learning-based denoiser in image
denoising are presented.
\end{abstract}

\section{Introduction}

\subsection{Problem statement}
Consider the classic problem of stochastic denoising: a stationary  process $\Xbbf=(X_i)_{i \in \ZZ}$, $X_i\in\Xc$,  is corrupted by additive white Gaussian noise $(Z_i)_{i \in \ZZ}$. Given $Y^n=X^n+Z^n$, the goal of a denoiser is to estimate $X^n$. In a Bayesian setting,  the full distribution of  $X^n$ is known. In such a setup, an estimator  that minimizes the  mean square error,   the MMSE  estimator, lets $\Xh^n=\E[X^n|Y^n]$. However, for general real-valued  sources with memory, computing  $\E[X^n|Y^n]$ is very demanding. Additionally, the actual denoising problem is more challenging, as in practice, we rarely have access to the source distribution. Instead, typically,  we are given either i)  some training dataset, or ii) some properties of the source extracted by  experts. 
 

Despite the mentioned challenges, there has been a large body of work on
developing efficient denoising algorithms. While for discrete sources,
especially for those with a small alphabet, there has been considerable
progress towards designing theoretically-founded asymptotically-optimal
efficient denoising algorithms, the same is not true for analog sources. In
analog denoising, most well-known methods are heuristic algorithms that are
developed using the expert knowledge for a specific type of data. For
instance, in image denoising, one of the classic approaches is based on
wavelet thresholding \cite{donoho1994ideal}, which works well because images
are known to be sparse in wavelet domain. (Refer to Chapter 11 of
\cite{stephane1999wavelet} for more information on wavelet denoising and its
extensions.) A more advanced denoising method, namely, BM3D
\cite{dabov2007image}, goes beyond sparsity in the wavelet transform and
enhances sparsity by looking for similar 2D patches in an image. In both
cases, the employed structure is discovered by image processing experts. The
denoising algorithm is tailored such that it takes advantage of that type of
structure. Other more recent techniques that achieve impressive results are
heuristic methods that employ neural networks
\cite{burger2012image,Ulyanov2018Deep}.

In this paper, inspired by recent progress in the Bayesian compressed sensing
of general analog sources \cite{jalali2018new}, we first propose, a
theoretically-founded Bayesian denoising method that is not tailored for a
specific source distribution. We refer to the new method quantized maximum a
priori (Q-MAP) denoiser, as it employs proper quantization of the source
alphabet to measure the conformation of a sequence with the known source
distribution. For memoryless sources and first-order Markov processes, we
characterize the asymptotic performance of the proposed method, as the
variance of the noise goes to zero, and the ambient dimension of the signal
grows to infinity. We show first-order optimality of the Q-MAP denoiser for
stationary memoryless sources. That is, for such sources, the limit of the
expected distortion achieved by the Q-MAP denoiser divided by the variance of
the noise converges to its optimal value achieved by an MMSE denoiser. We
also show that the proposed method leads to a learning-based denoising method
that can be applied to various data types. 
technique to image denoising and report some initial numerical results.

The Q-MAP denoiser is based on an optimization over the space of possible reconstruction sequences. Given  $y^n$ and the distribution of the source, the cost assigned to a  candidate sequence $u^n$ by the Q-MAP denoiser consists of the weighted addition of  1) $\norm{u^n-y^n}^2$, and 2) a term that  measures the conformity of $u^n$ with the source distribution. The novelty of Q-MAP denoiser lies in the second term. This term summarizes the source distribution into potentially exponentially many weights. However, as explained later, most of such terms are not important and indeed there are a relatively small number of such terms that need to be taken into account. Additionally, this property guides us to design a learning-based denoiser that ``learns'' those key weights form available training data. As an initial proof of concept, we use ImageNet database \cite{deng2009imagenet} to train our model, and explore the performance of the proposed method in image denoising.

The organization of the paper is as follows. Section \ref{sec:Q-MAP}
introduces the Q-MAP denoiser. Section \ref{sec:ID-MMSE} reviews the concept
of the information dimension of a process and its known connections with the
asymptotic MMSE behavior of an optimal Bayesian denoiser, for memoryless
sources. Section \ref{sec:main} presents the main theoretical results we
prove for the Q-MAP denoiser. Section \ref{sec:proofs} presents the detailed
proofs of the main results. Section \ref{sec:conclusion} concludes the paper.

%
%

\subsection{Notations and definitions}

Given a sequence $u^n$, with discrete alphabet $\Uc$, define its $k$-th order
empirical distribution as $\hat{p}^{k}(\cdot|u^n)$. That is, for
$a^{k}\in\Uc^{k}$,
\begin{align}
\hat{p}^{k}(a^{k})  = {1\over n-k+1} \sum_{i=1}^{n-k+1} \ind_{u_i^{i+k-1}=a^{k}}.\label{eq:emp-dist}
\end{align}

Given $x\in\mathds{R}$, the $b$-bit quantized version of $x$ is defined as
$\Quant{x}{b}\triangleq2^{-b}\lfloor 2^bx\rfloor$. For a vector
$x^k\in\mathds{R}^n$, $\Quant{x^k}{b}\triangleq (\Quant{x_1}{b},\ldots,
\Quant{x_k}{b}).$

For $x\in\mathds{R}$, $\delta_x$ denotes the Dirac measure with an atom at
$x$. $\log$ and $\ln$ refer to logarithm in base $2$ and normal logarithm,
respectively.

\section{Q-MAP denoiser}\label{sec:Q-MAP}

Consider  stationary process $\Xbbf=(X_i)_{i \in \ZZ}$ corrupted by process  $(Z_i)_{i \in \ZZ}$, where  $(Z_i)_{i \in \ZZ}$ are independently identically distributed (i.i.d.) $\Nc(0,\sigma^2)$. In Bayesian denoising setting, we observe  $Y^n = X^n +  Z^n$ and have access to the full distribution of the process $\Xbbf$. Inspired by the recovery method proposed in \cite{jalali2018new} for  Bayesian compressed sensing, we propose the following denoiser, which we refer to as Bayesian quantized maximum a priori  (Q-MAP) denoiser. First pick memory parameter $k\in\ZZ^+$ and quantization level $b\in\ZZ^+$. Let $\Xc$ and $\Xc_b$ denote the alphabet of  $\Xbbf$ and its $b$-bit quantized version, respectively. That is, 
\[
\Xc_b=\{\Quant{x}{b}: x\in\Xc\}.
\]
For  $a^{k+1}\in\Xc_b^{k+1}$, define weight $w_{a^{k+1}}>0$ as
\begin{align}
w_{a^{k+1}}\triangleq -\log \pi^k_b(a^{k+1}),\label{eq:weights}
\end{align}
where  $\pi_b^{k+1}$  denotes the distribution of  $\Quant{X^{k+1}}{b}$, \ie  
\begin{align}
\pi^k_b(a^{k+1})=P(\Quant{X^{k+1}}{b}=a^{k+1}). 
\end{align}
 Then, given $\lambda>0$,  the Q-MAP denoiser estimates $X^n$, as 
 \begin{align}
    \hat{X}^{n, k, b} = \argmin_{u^n\in\Xc^n} & \Big[\norm{Y^n - u^n}_2^2+  \lambda c_{w}(u^n)\Big],\label{eq:Q-map-d}
\end{align}
where
\[
c_w(u^n)\triangleq \sum_{a^{k+1}\in\Xc_b^{k+1}} w_{a^{k+1}}\hat{p}^{k+1}(a^{k+1}|\Quant{u^n}{b}).
\]
Here, $\hat{p}^{k+1}(\cdot|u^n)$ denotes the empirical distribution defined
earlier in \eqref{eq:emp-dist}. To understand how Q-MAP works, note that the
cost function consists of two terms. The first term $\norm{Y^n - u^n}_2^2$ is
a familiar measure that tries to ensure that the reconstruction sequence is
not far from the observed vector. On the other hand, the second term,
$c_w(u^n)$, is less familiar. The role of this term is to impose the
structure of the source and find a sequence that is consistent with the
source distribution. To better understand function $c_w(u^n)$, note that,
using some simple algebra, $c_w(u^n)$ can be written as
\begin{align} 
c_w(u^n)=\sum_{i=1}^{n-k} w_{\Quant{u^{i+k}_i}{b}}.
\end{align}
In other words, to compute the cost associated with a potential
reconstruction sequence $u^n$, one needs to slide over $u^n$ with a window of
length $k+1$. Each block of length $k+1$ is quantized by $b$ bits and then
the weight associated it is considered. All these weights, corresponding to
all quantized blocks of length $k+1$, are added together.

To better understand the Q-MAP denoiser, its associated optimization, and its implications, we next review two classic examples. 
\begin{example}[Sparse source]\label{ex:1}
Consider an i.i.d.~process $\Xbbf$, where $X_i\sim (1-q_0)\delta_0+q_0\pi_c$, $\pi_c$ denoting the pdf of a uniform distribution over $(0,1)$. It is straightforward to verify that, for $k=1$,
 \[
   \textstyle{ \hat{X}^{n, k, b} = \argmin_{u^n\in\Xc^n}  \Big[\norm{Y^n - u^n}_2^2+  \lambda(1+\gamma) \norm{u^n}_0 \Big]},
\]
  where $\gamma=O({1\over b})$ is an absolute constant only depending on $q_0$ and $b$.
Here, $\norm{u^n}_0=|\{i: u_i\neq 0\}$.
\end{example}

\begin{example}[Piecewise-constant 1-Markov source] \label{ex:2}As the next example, consider a stationary 1-Markov process, where conditioned on $X_i=x_i$, $X_{i+1}$ is distributed as $(1-q_0)\delta_{x_i}+q_0\pi_c$, where $\pi_c$ is the same pdf as the previous example. Let $N_{\rm jump}(u^n)=\{i: u_i\neq u_{i+1}\}$ count the number of jumps in $u^n$. Again it is straightforward to verify that, for $k=2$, 
\[
     \hat{X}^{n, k, b}=\argmin_{u^n\in\Xc^n}\Big[{1\over n}\norm{Y^n-u^n} +\lambda(1+\gamma)N_J(u^n)\Big],
  \]
  where $\gamma=O({1\over b})$ is an absolute constant only depending on $q_0$ and $b$.
\end{example}

  These examples  can  be generalized to much
  more general distributions. Moreover, while \eqref{eq:Q-map-d}
  might suggest that there are exponentially many weights that need to be
  computed (or learned), these examples verify that for structured processes,
  that number is substantially smaller than $|\Xc_b|^{bk}$. In fact, it can
  be seen that, those weights correspond to the key features of the
  distribution, e.g., sparsity or being piece-wise constant. After such
  simplification of the weights, one can consider different routes towards
  approximating the solution of \eqref{eq:Q-map-d} (for instance one based on
  dynamic programming). In Section \ref{sec:numerical}, we explain one such
  method and report some initial results.

\section{Information dimension  and MMSE}\label{sec:ID-MMSE}

In information theory, the entropy rate of a discrete stationary process is a
well-known measure of the level of information involved in that process
\cite{cover}. On the other hand, all analog processes have an infinite
entropy rate. However, this does not mean that all analog processes are the
same. This manifests itself in various applications, such as compressed
sensing, \ie recovering a source from its under-determined linear
measurements. Not all analog stationary processes lend themselves to compressed
sensing. Therefore, a measure of ``structuredness'' beyond entropy rate is
required for such analog sources. One such measure of structuredness  that in recent years have shown to be very relevant is R{\'e}nyi information dimension \cite{renyi1959dimension}.
\begin{definition}
The upper  R{\`e}nyi ID of random variable $X$ is defined as
\[
\bar{d}(X)=\limsup_{b\to\infty} {H([X]_b)\over b}.
\]
The lower R{\`e}nyi of $X$ is defined as $\underline{d}(X)=\liminf_{b\to\infty} {H([X]_b)\over b}$). If $\bar{d}(X)=\underline{d}(X)$, the R{\`e}nyi information dimension of random variable $X$ is defined as $d(X)=\bar{d}(X)=\underline{d}(X)$.
\end{definition}
This definition was generalized in \cite{JalaliP:17-IT} to stationary analog processes. 
\begin{definition}
The upper  ID of stationary process $\Xbbf$ is defined as
\[
\bar{d}(\Xbbf)=\lim_{k\to\infty} \limsup_{b\to\infty} {H(\Quant{X_{k+1}}{b}|\Quant{X^{k}}{b})\over b}, 
\]
The lower  ID of  process $\Xbbf$ is defined analogousely, by replacing $\limsup$ with $\liminf$. If $\bar{d}(\Xbbf)=\underline{d}(\Xbbf)$, the ID of stationary process $\Xbbf$ is defined as $d(\Xbbf)=\bar{d}(\Xbbf)=\underline{d}(\Xbbf)$.
\end{definition}
To gain some intuition on this measure, note that the ID of the i.i.d.~process studied in Example \ref{ex:1} can be shown to be equal to $q_0$ \cite{renyi1959dimension}. Moreover, the ID of the 1-Markov process described in \Cref{ex:2} is also equal to $q_0$ \cite{JalaliP:17-IT}.

To see the connection between the defined ID  and the denoising problem, consider the problem of scalar denoising: random variable $X$ is corrupted with an additive Gaussian noise $Z$,  $Z\sim\Nc(0,\sigma^2)$, as $Y=X+Z$.   Let $\Xh=\E[X|Y]$ denote the MMSE estimator of $X$ given $Y$ and the distribution of $X$.  It has been proven that    for discrete random variables and mixture of discrete and continuous random variables, $\lim_{\sigma\to 0}{1\over \sigma^2}\E[(X-\Xh)^2]$   is equal to the R{\'e}nyi ID of $X$ \cite{wu2011mmse}.  Roughly speaking, in the next section,  we prove that  for i) memoryless sources and ii) first order Markov sources, asymptotically, the proposed Q-MAP denoiser achieves the ID of the source. Combined with the mentioned results on the connection between MMSE and R{\`e}nyi ID, these results provide further evidence on the effectiveness of the proposed method. 

\section{Main Results}\label{sec:main}

We show that for appropriate choices of parameters, the Q-MAP estimator
$\hat{X}^{n, k, b}$ is able to recover structured signals in the high-SNR
($\sigma \rightarrow 0$) regime with the information-theoretic optimal
loss. More precisely, as an illustration of this general phenomenon,
we consider two cases: a first case where $X$ is i.i.d., and a more general 
case where $X$ is 1-Markov.

\begin{theorem}
    \label{thm:recovery-iid}
    Suppose that $\Xbbf$ denotes an i.i.d. process, and that the distribution of $X_0$ verifies:
    \begin{equation}
        X_0 \sim q_0 \pi_c + \sum_{p = 1}^m q_p \delta_{x_p},\label{eq:iid-dist}
    \end{equation}
    where $q_p$ are positive weights such that $\sum_{p = 0}^m q_p = 1$,
    $x_1, \dotsc, x_m \in (0, 1)$ denote arbitrary given values,
    and $\pi_c$ denotes an absolutely continuous distribution with bounded density.
    Then, for $\lambda = \sigma^{3 / 2}$, the Q-MAP estimator verifies:
    \begin{equation}
        \lim_{\sigma \rightarrow 0} \lim_{b \rightarrow \infty}
        \frac{1}{\sigma^2} \EE \left(X_0 -  \hat{X}_0^{1, 1, b}\right)^2 = q_0.
    \end{equation}
\end{theorem}

\begin{theorem}
    \label{thm:recovery-1-markov}
    Suppose that $X$ denotes a 1-Markov process, such that the conditional distribution $X_{i + 1} \mid X_i$
    verifies:
    \begin{equation}
        X_{i + 1} \mid X_i \sim q_0 \pi_c + \sum_{p = 1}^m q_p \delta_{f_p(X_i)},
    \end{equation}
    where $q_p$ are weights such that $\sum_{p = 0}^m q_p = 1$, and $f_p$ denote Lipschitz-continuous
    functions with a given constant $L$.

    Then, for $\lambda = \sigma^{3 / 2}$, the Q-MAP estimator verifies:
    \begin{equation}
        \lim_{n \rightarrow \infty} \lim_{\sigma \rightarrow 0} \lim_{b \rightarrow \infty}
        \frac{1}{\sigma^2} \EE \left(X_{\lfloor n / 2 \rfloor} - \hat{X}^{n, 2, b}_{\lfloor n / 2 \rfloor}\right)^2 = q_0.
    \end{equation}
\end{theorem}

We expect \cref{thm:recovery-iid} and \cref{thm:recovery-1-markov} to be instances of a more
general theorem which would show that the Q-MAP estimator recovers well-behaved signals
which are structured in the sense of the R\'enyi information dimension at the optimal
rate. We leave the exact statement and proof of such a theorem to future work.

\section{Proofs}\label{sec:proofs}
\subsection{I.I.D. setting}

 In general, we aim to show that for a high enough quantization level, it is
 possible to recover the ``structure'' of the signal with high probability in
 the high-SNR regime, and thus we only incur error on the ``unstructured''
 part of the signal.

To formalize the notion of structure in such a process, we define
discrete-valued process $\Sbbf$ as follows. For $i\in\mathds{Z}$,
$S_i\in{\cal S}\triangleq\{0,1,\ldots,m\}$. If $X_i=x_p$, $p=1,\ldots,m$,
then $S_i=p$. If the value of $X_i$ is drawn from $\pi_c$, $S_i=0$. In other
words, $\PP(X_i=x_p|S_i=p)=1$. Conditioned on $S_i=0$, $X_i$ is drawn
independently from $\pi_c$.

Note that since we have assumed that $k=1$, the denoiser simplifies to a symbol-by-symbol denoiser, where
\begin{equation}
    \hat{X}_i = \argmin_{u \in (0, 1)}[ (Y_i - u)^2 + \frac{\lambda}{b} c_w(u)],
\end{equation}
where by definition $c_w(u)= -\log \pi_b([u]_b)$. Here $\pi_b$ denotes the
distribution of $[X_i]_b$, where $X_i$ is distributed as \eqref{eq:iid-dist}.

\subsubsection{Regularizer}
To prove that the described denoiser is able to recover the structure of
$X^n$, \ie $S^n$, with small probability of error, we first focus on the
limiting behavior of the regularizer, as $b$ grows to infinity.
\begin{lemma}
    \label{lemma:convergence-regularizer-iid}
Let $r_b(u) \triangleq {1\over b}c_w(u)$.      
        For any $b > 0$, and any $ u \in (0, 1)$ with $[u]_b\neq [x_p]_b$, for all $p\in\{1,\ldots,m\}$, we have $r_b(u)\leq 1$ and 
    \begin{align*}
r_b(u) \geq 1 - C_1 2^{-b},
    \end{align*}
    where $C_1=\log\sup_x\pi_c(x)$. 

    Additionally, for any $b>0$, if  $ u \in (0, 1)$ satisfies $[u]_b=[x_p]_b$, for some $p \in \{1, \dotsc, m\}$, then
    \begin{align*}
r_b(u) \leq {C_2\over b},
    \end{align*}
    where $C_2 = \log( {1 \over  \min_{p} q_p })$.
\end{lemma}
In particular, we see that in the limit, $r_b(u)$ is the indicator of whether
the given value has non-zero probability under the prior distribution. This
is similar to $\ell_0$ regularization (or hard-thresholding) for sparse signals,
which is known to be optimal.

\begin{proof}
    Let $\QuantInt{u}{b}$ denote the interval $\QuantInt{u}{b} = [\Quant{u}{b}, \Quant{u}{b} + 2^{-b})$. By definition, 
    \begin{align*}
        \pi_b(\Quant{u}{b})
        &= \PP(X \in \QuantInt{u}{b})  \\
        &= q_0 \int_{\Quant{u}{b}}^{\Quant{u}{b} + 2^{-b}} \pi_c(x) \, dx 
 + \sum_{p = 1}^m q_p \ind_{x_p \in \QuantInt{u}{b}} \\
        &= 2^{-b} \pi_c(u^*) + \sum_{p = 1}^m q_p \ind_{x_p\in \QuantInt{u}{b}},
    \end{align*}
    where the last step holds by the   the mean value theorem for some $u^*\in \QuantInt{u}{b}$.
    
    First, assume that  $[u]_b\neq [x_p]_b$, for all $p\in\{1,\ldots,m\}$. 
    This implies that $\ind_{x_p\in \QuantInt{u}{b}}=0$, for all $p$. Therefore, 
    \begin{equation}
        r_b(u) =  \frac{b-\log \pi_c(u^*)}{b} \geq 1 - C_1 / b.
    \end{equation}
 On the other hand, assume that  $ u \in (0, 1)$ satisfies $[u]_b=[x_p]_b$, for some $p$.  Then, $\pi_b(\Quant{u}{b})=2^{-b} \pi_c(u^*)+q_p$. Therefore, 
    \begin{equation}
        r_b(u) \leq - \frac{1}{b} \log( 2^{-b}\pi_c(u^*) + q_p) \leq \frac{\log( {1 \over  q_p})}{b}\leq {C_2\over b}.
    \end{equation}
\end{proof}

\subsubsection{Structure Estimation}

Given the previous remark on the regularizer, we show that this is enough
for the estimator $\hat{X}$ to recover the structure $S$. More precisely,
we have the following result. In the following, we have dropped subscript
$i$, and refer to $X_i$, $Y_i$ and $S_i$, as $X$, $Y$, $S$, respectively.
Also, in the following, let
\[
d(u) \triangleq \min_{p\in\{1,\ldots,m\}} \abs{u - x_p}.
\] 
That is, $d(u)$ measures the  distance between $u$ and closest singularity point of the $\pi$. 
In the next lemma, $C_2$ refers to the constant defined in \cref{lemma:convergence-regularizer-iid}.
\begin{lemma}    \label{lemma:structure-estimation-iid}
  Define events
    \begin{align}
    \Ec_1 = \{ d(Y) \geq \sqrt{\lambda} + 2^{-b}\},\label{eq:def-E1}
    \end{align}
    and
       \begin{align}
    \Ec_2 = \{ \abs{Z} \leq \sigma \log (1 / \sigma) \}.\label{eq:def-E2}
    \end{align}
    Then, conditioned on $\Ec_1\cap\{ S = 0 \}$, $d(\hat{X}) \geq 2^{-b}$, and
    $\abs{Y - \hat{X}} \leq \sqrt{C_2 / b}$. Additionally, conditioned on $\Ec_1^c \cap \{ S = 0 \}$,
    we have $\abs{Y - \hat{X}}^2 \leq \lambda +\gamma$, where $\gamma=O(b^{-1})$. On the other hand, conditioned on
    $\Ec_2\cap \{S=p\}$,
    \begin{align*}
        \abs{\hat{X} - x_p} \leq 2^{-b}.
    \end{align*}
\end{lemma}

\begin{proof}

    Let $L(u) = (Y - u)^2 + \lambda r_b(u)$ and $\Xh=\argmin_u L(u)$. 
    First, suppose that $S = 0$. 
By the triangle inequality, for any $p\in\{1,\ldots,m\}$, 
    \begin{equation}
       \sqrt{ L(\hat{X})} \geq \abs{Y - \hat{X}} \geq \abs{Y - x_p} - \abs{x_p - \hat{X}}.
    \end{equation}
But,  $L(\hat{X}) \leq L(Y)=\lambda r_b(Y) \leq \lambda.$ Therefore, 
\begin{align}
\abs{Y - x_p} - \abs{x_p - \hat{X}}\leq \sqrt{\lambda}.\label{eq:Y-xp-dist}
\end{align}
    Conditioned on $\Ec_1\cap\{S=0\}$, from \eqref{eq:Y-xp-dist}, $\abs{x_p - \hat{X}}\geq \sqrt{\lambda} + 2^{-b}-\sqrt{\lambda}=2^{-b}$. 
Additionally, this implies by \cref{lemma:convergence-regularizer-iid} that it suffices to optimize over
    values of $u$ for which $r_b(u) \geq 1 - C_2 / b$. Thus, we have that $L(u) \geq (Y - u)^2 + \lambda(1 - C_2 / b)$. But,
    $L(Y) \leq \lambda$, from which we deduce that $(Y - \hat{X})^2 \leq \lambda C_2 / b$.

    Additionally, note that on the event $\Ec_1^c$, we have $\abs{Y - x_p} \leq \sqrt{\lambda} + 2^{-b}$, for
    some $p \in \{1, \dotsc, k\}$. Now,  by \cref{lemma:convergence-regularizer-iid},
    $r_b(x_p) \leq C_2 / b$. Therefore, conditioned on $\Ec_1^c\cap\{S=0\}$, 
    \begin{align*}
    L(x_p) &\leq (Y - x_p)^2 + {\lambda C_2 \over  b}\nonumber\\
    &\leq  ( \sqrt{\lambda}  + 2^{-b})^2 + {\lambda C_2 \over b}.
   \end{align*}
    By optimality of $\hat{X}$, we have that $L(\hat{X}) \leq L(x_p)$, from which we conclude that
    \begin{equation*}
        (Y - \hat{X})^2 \leq ( \sqrt{\lambda}  + 2^{-b})^2 + {\lambda C_2 \over  b}=\lambda+\gamma.
    \end{equation*}
    where $\gamma=2^{-b+1}\sqrt{\lambda}  + 2^{-2b} + {\lambda C_2 \over  b}=O(1/b)$.

    To prove the next part, suppose that $S \neq 0$, and suppose by contradiction that $d(\hat{X}) >  2^{-b}$.
    Then by \cref{lemma:convergence-regularizer-iid},  $r_b(\hat{X}) > 1 - C_1 / b$,
    and therefore, $L(\hat{X}) \geq \lambda (1 - C_1 / b)$.
 
 Additionally, since $S\neq 0$,  there exists $\tilde{X}$ such that $\abs{X - \tilde{X}} < 2^{-b}$, and 
 by \cref{lemma:convergence-regularizer-iid}, $r_b(\tilde{X}) < C_2 / b$. Conditioned on $\Ec_2$, 
     \begin{equation*}
        L(\tilde{X}) \leq (\sigma \log (1 / \sigma) + 2^{-b})^2 + \lambda C_2 / b.
    \end{equation*}
Now, note that $\lambda = \sigma^{3 / 2} \geq \sigma^2 \log(1 / \sigma)$ for $\sigma$ small.
We thus have that $L(\tilde{X}) < L(\hat{X})$ for $b$ large enough, which is a contradiction.
\end{proof}

\subsubsection{Estimation}

Finally, we may combine the previous result with a conditional analysis of the error
(depending on whether $S = 0$) to obtain the final result.
\begin{proof}[Proof of \Cref{thm:recovery-iid}]
First, define events $\Ec_1$ and $\Ec_2$ as \eqref{eq:def-E1} and
\eqref{eq:def-E2}, respectively. Note that, since for $S=0$, $Y=X+Z$ has a
continuous distribution, $\PP(\Ec_1) \rightarrow \PP(S=0)=q_0$. Also,
$\PP(\Ec_2^c) \rightarrow 0$, as $\sigma \rightarrow 0$.

First, note that
    \begin{align*}
        \EE[(X - \hat{X})^2]
        &= \EE[(X - \hat{X})^2 \mid S = 0] \PP(S = 0) \\
        &\quad + \sum_{p = 1}^m \EE[(X - \hat{X})^2 \mid S = p] \PP(S = p).
    \end{align*}
For $S \neq 0$, from \cref{lemma:structure-estimation-iid},
    \begin{equation*}
        \EE[(X - \hat{X})^2 \mid S = p] \leq  2^{-2b} +  \PP(\Ec_2^c),
    \end{equation*}
    where the last line follows because $X$ is bounded by one. Now, $\PP(\Ec_2^c) \leq 2 \sigma^{-\log \sigma / 2}$ by
    the Gaussian tail bound, and in particular we have $\frac{1}{\sigma^2}\EE[(X - \hat{X})^2 \mid S = p] \rightarrow 0$
    as $b \rightarrow \infty$ and $\sigma \rightarrow 0$.

    On the other hand, 
    \begin{align*}
      &  \EE[(X - \hat{X})^2 \mid S = 0]
        \leq \EE[(\abs{X - Y} + \abs{Y - \hat{X}})^2 \mid S = 0] \\
        &\leq \EE[(X - Y)^2 \mid S = 0]  + 2\EE[(\abs{X - Y}\abs{Y - \hat{X}}) \mid S = 0] \\
        &\quad + \EE[(Y - \hat{X})^2 \mid S = 0].
    \end{align*}
Since $Y=X+Z$, $\EE[(X - Y)^2 \mid S = 0] =\EE[Z^2]= \sigma^2$.     Also, 
    \begin{align*}
 &\EE[(\abs{X - Y}\abs{Y - \hat{X}}) \mid S = 0]  \\
 &= \EE[(\abs{X - Y}\abs{Y - \hat{X}}) \mid S = 0,\Ec_1]\PP(\Ec_1\mid S=0) \\
 &\quad +\EE[(\abs{X - Y}\abs{Y - \hat{X}}) \mid S = 0,\Ec_1^c]\PP(\Ec_1^c\mid S=0)
\end{align*}
From \cref{lemma:structure-estimation-iid}, conditioned on
$\Ec_1\cap \{ S = 0 \}$, we have $\abs{\hat{X} - Y} \leq \sqrt{C_2 / b}$,
and conditioned on $\Ec_1^c\cap\{S=0\}$, we have $\abs{\hat{X} - Y} \leq \sqrt{\lambda + \gamma}$, where $\gamma=O(1/b)$.
Therefore,
\begin{align}
&\EE[(\abs{X - Y}\abs{Y - \hat{X}}) \mid S = 0]\nonumber\\
&\leq \sqrt{C_2 \over b} \EE[\abs{X - Y} \mid S = 0,\Ec_1]\PP(\Ec_1\mid S=0)\nonumber\\
&+ \sqrt{\lambda +\gamma}\; \EE[(\abs{X - Y}) \mid S = 0,\Ec_1^c]\PP(\Ec_1^c\mid S=0).\label{eq:error-two-terms-multiplied}
\end{align}
Now, note that conditioned on $S = 0$, $Y = X + Z$ has a continuous density which is uniformly
upper bounded for any $\sigma > 0$ (as $X$ is continuous on the event $S = 0$), and hence, by the
mean value theorem, we have that 
\begin{align}
\PP(\Ec_1^c \mid S = 0) \leq (\sqrt{\lambda} + 2^{-b})M,\label{eq:bound-P-E1c-S0}
\end{align}
for some $M > 0$ .
In particular, we  conclude that:
\begin{align}
&\EE[(\abs{X - Y}\abs{Y - \hat{X}}) \mid S = 0] \nonumber \\
&\leq \sqrt{C_2 \over b} \EE[\abs{X - Y} \mid S=0)
+ M(\lambda + \gamma_1) \EE[\abs{X - Y} \mid S = 0, \Ec_1^c].
\end{align}
where $\gamma_1 = o(b)$. Now, note that we have:
\begin{equation*}
    \EE[\abs{X - Y} \mid S = 0, \Ec_1^c] = \sigma \EE[\abs{Z} \mid S = 0, \Ec_1^c].
\end{equation*}
As $X$ has continuous distribution on the event $S = 0$, the event $\Ec_1^c$ does
not substantially affect the distribution of $Z$. For illustration, we give a proof
when there is a single singularity point $x_1$, in which case $\Ec_1^c = \{ \abs{Y - x_1} \geq \sqrt{\lambda} + 2^{-b} \}$.
Note that for any $\delta > 0$, we have:
\begin{equation*}
    \pi(z \mid \abs{X + \sigma Z - x_1} \leq \delta) =
    \frac{\pi(z) \int_{-\delta - x_1 - \sigma z}^{\delta - x_1 - \sigma z} \pi(x) \, dx}{\PP(\abs{X + \sigma Z - x_1} \leq \delta)}.
\end{equation*}
By the mean value theorem, we have that
\begin{gather*}
\int_{-\delta - x_1 - \sigma z}^{\delta - x_1 - \sigma z} \pi(x) \, dx = 2 \delta \pi(x^*)
\end{gather*}
where $x^* \in (-\delta - x_1 - \sigma z, \delta - x_1 - \sigma z)$. 
Additionally, note that $Y - x_1$ has a continuous
density uniformly bounded above and below for any $\sigma$, and hence there exists $m, M > 0$ such that
\begin{equation*}
    2m \delta \leq \PP(\abs{X + \sigma Z - x_1} \leq \delta) \leq 2M\delta.
\end{equation*}
We thus deduce that there exists $m, M > 0$, such that for any $\sigma$:
\begin{equation*}
    m\pi(z) \leq \pi(z \mid \abs{X + \sigma Z - x_1} \leq \delta) \leq M \pi(z).
\end{equation*}
In particular, we have that $\EE[\abs{Z} \mid S = 0, \Ec_1^c] \leq M \EE[\abs{Z}] \leq M\sigma \sqrt{2 / \pi}$. Substituting this in our previous result, it follows that
\begin{equation*}
\EE[(\abs{X - Y}\abs{Y - \hat{X}}) \mid S = 0] \leq \sigma\left[\sqrt{2C_2 \over \pi b} + M^2(\lambda + \gamma_1) \sqrt{2 \over \pi}\right],
\end{equation*}
where $\gamma_1=o(b)$.

Following a similar argument as the one used in bounding $\EE[(\abs{X - Y}\abs{Y - \hat{X}}) \mid S = 0]$, we have
\begin{align}
\EE[(\abs{Y - \hat{X}}^2)& \mid S = 0]\leq {C_2 \over b} +(\lambda +\gamma) \PP(\Ec_1^c\mid S=0)\nonumber\\
&\leq {C_2 \over b} +M(\lambda +\gamma) (\sqrt{\lambda} + 2^{-b}),
\end{align}
where the last line follows from \eqref{eq:bound-P-E1c-S0}. 
    Putting all the terms together, it follows that
    \begin{align}
        \EE[(X - \hat{X})^2 \mid S = 0]& \leq \sigma^2+ 2\sigma\left[\sqrt{2C_2 \over \pi b} + 2\sqrt{2 \over \pi}M^2(\lambda + \gamma_1) \right] \nonumber\\ 
&\quad+{C_2 \over b} +M(\lambda +\gamma) (\sqrt{\lambda} + 2^{-b}).
    \end{align}

    In particular, taking the limit $b \rightarrow \infty$, we have that
    $\EE[(X - \hat{X})^2 \mid S = 0] \rightarrow \sigma^2 + M\lambda^{3\over 2}+  \sqrt{2 \over  \pi} M^2\sigma \lambda$. Finally, since $\lambda=\sigma^{3/2}$, we have
    \begin{align*}
&        \lim_{\sigma \rightarrow 0} \lim_{b \rightarrow \infty} \frac{1}{\sigma^2}\EE[(X - \hat{X})^2]\\
        &= \lim_{\sigma \rightarrow 0, \ b \rightarrow \infty} \frac{q_0}{\sigma^2} \EE[(X - \hat{X})^2 \mid S = 0] \\
        &\;\;\;\;\;\; + (1 - q_0) \lim_{\sigma \rightarrow 0, \, b \rightarrow \infty} \frac{1}{\sigma^2}\EE[(X - \hat{X})^2 \mid S \neq 0] \\
        &= q_0.
    \end{align*}
\end{proof}

\subsection{1-Markov Setting}

The 1-Markov setting shares similarities with the i.i.d. setting, but is significantly more
complex as the estimator is no longer separable. Nevertheless, the basic intuition remains
similar, in that we will show that the estimator is able to capture the ``structure'' with
high probability, and given that, it is able to recover the signal at the claimed rate.

In a similar fashion to the i.i.d. case, we define process $\Sbbf$ taking values in $\{ 0, \dotsc, m \}$ to capture the structure of process $\Xbbf$. $S_i=p$, if $X_{i+1}=f_p(X_i)$.  If $X_{i+1}$ is drawn from $\pi^c$, independent of $X_i$, then $S_i=0$. Therefore,  $\PP(S_i = p) = q_p$.  
In this section,  let $\QuantInt{u}{b}$ denote the interval $[\Quant{u}{b}, \Quant{u}{b} + 2^{-b})$.

\subsection{Regularizer}

The regularizer presents properties similar to the i.i.d. case, but is complicated
by the fact that the quantization need not match between the domain and range of
the functions which induce the structure. We thus have the following lemma.

First, we consider a basic lemma on the structure of the regularizer.
\begin{lemma}[Convergence of Regularizer]
    \label{lemma:convergence-regularizer-markov}
    Let $\diststru(u_1, u_2) = \min_{p \in \{1, \dotsc, k\}} \abs{u_2 - f_p(u_1)}$.

    For any bitrate $b$, for  $u_1, u_2 \in (0, 1)$ with
    $\diststru(u_1, u_2) \geq C2^{-b}$, we have:
    \begin{equation}
        \regquant(u_1, u_2) \geq 2 - C / b
    \end{equation}
    where $C$ is an absolute constant.
    Additionally, for any $b$, $u_1 \in (0, 1)$ and $p \in \{1, \dotsc p\}$, there
    exists $u_2 \in (0, 1)$ such that $\abs{u_2 - f_p(u_1)} \leq C2^{-b}$ and
    \begin{equation}
        \regquant(u_1, u_2) \leq 1 + C / b.
    \end{equation}
    
    Finally,  for any bitrate $b$, and for any $u_1, u_2 \in (0, 1)$, we have:
    \begin{equation}
        1 - C / b \leq \regquant(u_1, u_2) \leq 2 + C / b.
    \end{equation}
\end{lemma}

\begin{proof}
By definition, 
\begin{align*}
    -b r_b(u_1, u_2)
    &= \log \PP(X_1 \in \QuantInt{u_1}{b}, X_2 \in \QuantInt{u_2}{b}) \\
    &= \log \sum_{p = 0}^k \PP(X_1 \in \QuantInt{u_1}{b}, X_2 \in \QuantInt{u_2}{b}, S_1 = p).
\end{align*}
Note that conditioned on $S_1 =0$, $X_1$ and $X_2$ are independent, and hence:
\begin{align*}
    &\PP(X_1 \in \QuantInt{u_1}{b}, X_2 \in \QuantInt{u_2}{b}, S_1 = 0) \\
    &= \PP(X_1 \in \QuantInt{u_1}{b}, X_2 \in \QuantInt{u_2}{b} \mid S_1 = 0) \PP(S_1 = 0) \\
    &= \PP(X_1 \in \QuantInt{u_1}{b} \mid S_1 = 0) \PP(X_2 \in \QuantInt{u_2}{b} \mid S_1 = 0) \PP(S_1 = 0).
\end{align*}
By the mean value theorem, $\PP(X_1 \in \QuantInt{u_1}{b})$ can be bounded as 
\begin{equation*}
    m 2^{-b} \leq \PP(X_1 \in \QuantInt{u_1}{b}) \leq M 2^{-b},
\end{equation*}
where $m, M$ denote the minimum and maximum of the density of
the stationary distribution.
Additionally, note that by the Lipschitz-continuity  assumption, 
\begin{equation}
    f_p(\QuantInt{u_1}{b}) \subseteq (f_p(u_1) - L 2^{-b}, f_p(u_1) + L 2^{-b}).
\end{equation}
In particular, we have that if $\abs{f_p(u_1) - u_2} \geq (L + 1)2^{-b}$, then
\[
\PP(X_1 \in \QuantInt{u_1}{b}, X_2 \in \QuantInt{u_2}{b}, S_1 = p) = 0.
\]
We thus deduce that, if $\abs{f_p(u_1) - u_2} \geq (L + 1)2^{-b}$ for all
$p$, then the events are independent, and hence:
\begin{align*}
    &\PP(X_1 \in \QuantInt{u_1}{b}, X_2 \in \QuantInt{u_2}{b}) \\
    &= \PP(X_1 \in \QuantInt{u_1}{b})\PP(X_2 \in \QuantInt{u_2}{b}) \\
    &\leq M^2 2^{-2b}.
\end{align*}
Substituting back into the definition of $r_b$, we have in that case that:
\begin{equation*}
    r_b(u_1, u_2) \geq 2 - 2b^{-1}\log M.
\end{equation*}
On the other hand, again by the Lipschitz assumption, 
\begin{align*}
    \QuantInt{u_1}{b}
    &\subseteq f_p^{-1} (f_p(u_1) - L 2^{-b}, f_p(u_1) + L 2^{-b}) \\
    &\subseteq \bigcup_{i = -L}^{L} f_p^{-1} \left(\dbracket*{f_p(u_1) + i 2^{-b}}_{b}\right)
\end{align*}
In particular, this implies that there exists an $i^*$ such that:
\begin{equation*}
    \abs*{\QuantInt{u_1}{b} \cap f_p^{-1} \left(\dbracket*{f_p(u_1) + i^* 2^{-b}}_{b}\right)} \geq \frac{2^{-b}}{2L + 1}.
\end{equation*}
Choose $u_2 = f_p(u_1) + i^* 2^{-b}$, then we have that:
\begin{align*}
    \PP&(X_1 \in \QuantInt{u_1}{b}, X_2 \in \QuantInt{u_2}{b}, S^1_p)\\
    &\geq \PP(X_1 \in f_p^{-1}(\QuantInt{u_2}{b}) \cap \QuantInt{u_1}{b}) \\
    &\geq \frac{m 2^{-b}}{2L + 1}.
\end{align*}
We deduce that:
\begin{align*}
    \regquant(u_1, u_2)
    &= - \frac{1}{b} \log \PP(X_1 \in \QuantInt{u_1}{b}, X_2 \in \QuantInt{u_2}{b}) \\
    &\leq 1 + \frac{1}{b}\abs*{\log \frac{m}{2L + 1}}.
\end{align*}

\end{proof}

\subsection{Structure Estimation}

As the regularizer is related to a structure indicator, we expect our
estimator $\hat{X}^n$ to reflect the structure of the data in the low noise regime.
More precisely, we have the following lemma.
\begin{lemma}
    \label{lemma:structure-recovery-markov}
    Let $d(u_1, u_2) = \min_{p} \abs{u_2 - f_p(u_1)}$.
    Consider events:
     \begin{align*}
        \Ec_1 &= \{d(Y_i, Y_{i + 1}) \geq \sqrt{2\lambda} + 2^{-b} \text{ or } S_i \neq 0, i=1,\ldots,n \}, \\
        \Ec_2 &= \{ \abs{\sum_{i = 1}^n Z_i} \leq n\sigma \log (1 / \sigma) \}.
    \end{align*}
    Then, conditioned on   $\Ec_1 \cap \Ec_2$, we have:
    \begin{align*}
        S_i = 0 &\Rightarrow d(\hat{X}_{i + 1}, X_i) \geq 2^{-b}, \\
        S_i \neq 0 &\Rightarrow \abs{\hat{X}_{i + 1} - f_{S_i}(\hat{X}_i)} \leq C2^{-b}.
    \end{align*}
\end{lemma}

\begin{proof}

We first show that if the estimated signal cannot have missing breaks, that is, if $S_i = 0$,
then this must be reflected in the estimated signal $\hat{X}^n$. We proceed by contradiction.
Suppose there exists $i_1 < i_2$ such that $\hat{X}^n$  verifies $\diststru(\hat{X}_i, \hat{X}_{i + 1}) \leq 2^{-b}$
for all $i \in \{i_1, \dotsc, i_2 - 1\}$.  Additionally, suppose that there exists
$l_1, \dotsc, l_m$ with $i_1 < l_1 < \dotsb < l_m < i_2$ and $S_{l_q} = 0$ for $q = 1, \dotsc, m$.
We define a new estimator:
\begin{equation}
    \tilde{X}_i = \begin{cases}
        \hat{X}_i &\text{ if } i \not\in \{ i_1 + 1, \dotsc, i_2 \}, \\
        \tilde{F}(Y_{l_p + 1}^{l_{p + 1}}, S_{l_p + 1}^{l_{p + 1}})_i &\text{ if } l_p + 1 \leq i \leq l_{p + 1},
    \end{cases}
\end{equation}
where we have defined the quantization-compatible analogue structure-oracle estimator
$\tilde{F}_b(Y, s)$ for a sequence
$s$ such that $s_i \neq 0$ as:
\begin{gather*}
    \tilde{F}_b(Y, s) = \argmin_{U \in \calC(s)} \norm{Y - U}_2^2, \\
    \calC(s) = \{ u : s_i \neq 0 \Rightarrow \abs{u_{i + 1} - f_{s_i}(u_i)} \leq 2^{-b} \\
     \text{ and } r_b(u_i, u_{i + 1}) < 1 + C / b \}.
\end{gather*}

Note that by \cref{lemma:convergence-regularizer-markov}, there exists
a $\check{X}^n \in \calC(s)$ such that $\norm{X^n - \check{X}^n} \leq C_2 2^{-b}$. In particular,
$\tilde{F}_b(Y, s)$ is well-defined as $\calC(s)$ is not empty.
By the definition of $\tilde{F}_b(Y, s)$ and the triangular inequality, we then have that
(using the fact that we are on the event $\Ec_2$):
\begin{align*}
    \norm{Y^n - \tilde{F}_b(Y, s)}_2
    &\leq \norm{Y^n - \check{X}^n}_2 \\
    &\leq C_2 2^{-b} + \norm{Y^n - X^n}_2 \\
    &\leq C_2 2^{-b} + n \sigma \log (1 / \sigma).
\end{align*}
Additionally, note that by the definition of $\calC(s)$, we have:
\begin{equation}
    r_b(\tilde{F}(Y, s)) \leq n (1 + C_1 / b).
\end{equation}

We may thus estimate the difference in objective between the proposed $\tilde{X}^n$
and the postulated $\hat{X}^n$. Note that they coincide outside of the interval
$i_1 + 1, \dotsc, i_2$, and hence suffices to consider that interval.
We have that:
\begin{align*}
    L(\tilde{X}^n) - L(\hat{X}^n)
    &\leq \sum_{q = 0}^k \norm{Y_{l_q + 1}^{l_{q + 1}} - \tilde{F}_b(Y_{l_q + 1}^{l_{q + 1}}, S_{l_q + 1}^{l_{p + 1}})}_2^2 \\
    &\quad - \norm{Y^n - \hat{X}^n}_2^2 + k \lambda (1 + C_4 / b) \\
    &\leq C (2^{-b} + \sigma^2 \log^2 (1 / \sigma)) - \norm{Y^n - \hat{X}^n}_2^2 \\
    &\quad + k\lambda (1 + C_4 / b).
\end{align*}

Now, note that we have assumed that $\hat{X}^n$ is structured at $l_q, q = 1, \dotsc k$,
that is, we have $\abs{\hat{X}_{l_q + 1} - f_{s_q}(\hat{X}_{l_q})} \leq 2^{-b}$ for
some choices of $s_q$. However, on the event $\Ec_1$, we note that $d(Y_i, Y_{i + 1}) > \sqrt{2\lambda} + 2^{-b}$,
and hence we deduce that:
\begin{equation*}
    (Y_{l_q} - \hat{X}_{l_q})^2 + (Y_{l_q + 1} - \hat{X}_{l_q + 1})^2 \geq 2\lambda.
\end{equation*}
Substituting into the previous result, we have that:
\begin{equation*}
    L(\tilde{X}^n) - L(\hat{X}^n) \leq C(2^{-b} + \sigma^2 \log^2(1 / \sigma))
    + k \{ -2\lambda + \lambda(1 + C / b) \},
\end{equation*}
which is negative for $b$ large and $\sigma$ small, a contradiction
of the optimality of $\hat{X}^n$.

We must also show the converse result, that our estimated signal $\hat{X}^n$ does
not have extraneous breaks. As previously,
we proceed by contradiction. Suppose that the condition fails at $i^*$, that
is, we have $d(\hat{X}_{i^*}, \hat{X}_{i^* + 1}) > c_1 2^{-b}$,
but there exists $i_1 < i^* < i_2$ such that $S_i \neq 0$ for all $i_1 \leq i < i_2$.

We define a new estimator $\tilde{X}^n$ to coincide with $\hat{X}^n$ outside of the
interval from $i_1$ to $i_2$, and to coincide with $\tilde{F}(Y_{i_1}^{i_2}, s(X)_{i_1}^{i_2})$
on the interval from $i_1$ to $i_2$. We may again compute the difference in loss as (using
the same results as above):
\begin{align*}
    L(\tilde{X}) - L(\hat{X})
    &\leq \norm{Y_{i_1}^{i_2} - \tilde{F}_b(Y_{i_1}^{i_2}, S_{i_1}^{i_2})}_2^2 - \lambda (1 - C_1 / b) \\
    &\leq C (2^{-b} + \sigma^2 \log^2(1 / \sigma)) - \lambda (1 - C_1 / b).
\end{align*}
Now, for $b$ large enough, and noting that $\lambda / (\sigma^2 \log^2(1 / \sigma)) \rightarrow \infty$,
we have that $L(\tilde{X}^n) - L(\hat{X}^n) < 0$, a contradiction.

\end{proof}

\subsection{Response Estimation}

Given the structure of the estimator established in the previous section,
we may show that the estimator achieves the claimed performance. As the estimation problem
is more complex in this case, we first consider the case with no quantization. In this case,
we have the following result.
\begin{lemma}
    \label{lemma:mse-analogue}
    Suppose that we observe a signal $Y^n$ such that:
    \begin{equation*}
        Y_i = g_i(\theta) + \sigma\epsilon_i,
    \end{equation*}
    where $g_i$ are known smooth functions, and $\epsilon_i$
    denotes i.i.d. gaussian noise.

    Let $\hat{\theta}$ denote the maximum likelihood estimator given by:
    \begin{equation*}
        \hat{\theta} = \argmin_u \sum_{i = 1}^n (Y_i - g_i(u))^2.
    \end{equation*}
    Then, we have that:
    \begin{equation*}
        \lim_{\sigma \rightarrow 0} \frac{1}{\sigma^2}\EE \sum_{i = 1}^n (g_i(\hat{\theta}) - g_i(\theta))^2 = 1.
    \end{equation*}
\end{lemma}

\begin{proof}

Let $L(u) = \sum_{i = 1}^n (Y_i - g_i(u))^2$ denote the loss function. Note that
we may equivalently write:
\begin{align*}
    L(u)
    &= \sum_{i = 1}^n (g_i(\theta) - g_i(u) + \sigma \epsilon_i)^2 \\
    &= \sum_{i = 1}^n (g_i(\theta) - g_i(u))^2 + 2 \sigma \sum_{i = 1}^n \epsilon_i (g_i(\theta) - g_i(u)) + C \\
    &= \sum_{i = 1}^n [g'_i(\theta)]^2(\theta - u)^2 + 2 \sigma \sum_{i = 1}^n \epsilon_i g'_i(\theta)(\theta - u) \\
    &\quad + R(\theta - u) + C,
\end{align*}
where $C$ denote a constant which does not depend on $u$, and $R(u)$ denote the remainder
of the Taylor expansion of $L(u)$ about $\theta$. Now, we may directly differentiate $L(u)$
to obtain a first-order condition for $\hat{\theta}$, namely:
\begin{equation*}
    \sum_{i = 1}^n [g'_i(\theta)]^2 (\theta - \hat{\theta}) = \sigma \sum_{i = 1}^n \epsilon_i g'_i(\theta) + R(\theta - \hat{\theta}),
\end{equation*}
which we may re-arrange as:
\begin{equation*}
    \frac{\theta - \hat{\theta}}{\sigma} = \frac{\sum_{i = 1}^n \epsilon_i g'_i(\theta)}{\sum_{i = 1}^n [g'_i(\theta)]^2} + \frac{R(\theta - \hat{\theta})}{\sigma}.
\end{equation*}

Using the Lagrange form of the remainder, we have that:
\begin{equation*}
    R(\theta - \hat{\theta}) = (\theta - \hat{\theta})^2 \sum_{i = 1}^n g''_i(\theta^*),
\end{equation*}
and hence we have:
\begin{equation*}
    \frac{\theta - \hat{\theta}}{\sigma} \left(1 + (\theta - \hat{\theta})\sum_{i = 1}^n g''_i(\theta) \right) = \frac{\sum_{i = 1}^n \epsilon_i g'_i(\theta)}{\sum_{i = 1}^n [g'_i(\theta)]^2}.
\end{equation*}
Now, we have that $\hat{\theta} \rightarrow \theta$ as $\sigma \rightarrow 0$, and hence we may take limits in the equation
above to obtain that:
\begin{equation*}
    \lim_{\sigma \rightarrow 0} \frac{\theta - \hat{\theta}}{\sigma} = \frac{\sum_{i = 1}^n \epsilon_i g'_i(\theta)}{\sum_{i = 1}^n [g'_i(\theta)]^2}.
\end{equation*}
Plugging this into the loss, and making use of a similar Taylor expansion, we obtain the claimed result.

\end{proof}

Finally, we use the previous lemmas to prove the main theorem.
\begin{proof}[Proof of \Cref{thm:recovery-1-markov}]
We first consider an auxiliary quantity. Let $\tilde{X}^n$ be defined as:
\begin{align*}
    \tilde{X}^n &= \argmin_{U^n \in \Dc(S)} \norm{Y^n - U^n}_2^2, \\
    \Dc(S) &=  \{ U^n: S_i=0\;{\rm or} \; U_{i + 1} = f_{S_i}(U_i),  i=1,\ldots,n\}.
\end{align*}
By \cref{lemma:mse-analogue}, we have that:
\begin{equation*}
    \lim_{\sigma \rightarrow 0} \EE[(\tilde{X}_{n / 2} - X_{n / 2})^2 \mid S] = \frac{1}{L \wedge n},
\end{equation*}
where $L$ denotes the length of the structure piece which contains the $n / 2$ index. More precisely, $L$ is such that:
\begin{gather*}
    S_m = S_{m + L} = 0, \\
    S_{m + 1} \neq 0, \dotsc, S_{n / 2} \neq 0, \dotsc, S_{m + L - 1} \neq 0.
\end{gather*}
As $S$ is an i.i.d. process, we may compute the distribution of $L$ explicitly, and in particular,
we have that $\EE [1 / L] = q_0$. We thus deduce that, for our auxiliary oracle estimator $\tilde{X}$:
\begin{equation*}
    \lim_{n \rightarrow \infty} \lim_{\sigma \rightarrow 0} \EE(\tilde{X}_{n / 2} - X_{n / 2})^2 = q_0.
\end{equation*}

Now, let $m = \abs{\{ i : S_i = 0 \}}$, and note that by \cref{lemma:structure-recovery-markov}, we have that
$r_b(\hat{X}) \geq m (2 - C / b) + (n - m) (1 - C / b)$. On the other hand, by \cref{lemma:convergence-regularizer-markov},
and the definition of $\Dc(S)$, we have that there exists $\check{X}$ such that $\norm{\check{X} - \tilde{X}}_2 \leq C 2^{-b}$,
and $r_b(\check{X}^n) \leq m (2 - C / b) + (n - m)(1 + C / b)$. Now, by definition of $\hat{X}^n$, we have that $L(\hat{X}^n) \leq L(\check{X}^n)$,
from which we deduce:
\begin{align*}
    \norm{Y^n - \hat{X}^n}_2^2 &\leq \norm{Y^n - \check{X}^n}_2^2 + \lambda n C / b \\
&    \leq \norm{Y^n - \tilde{X}^n} + C 2^{-b} + \lambda n C / b.
\end{align*}
Now again by \cref{lemma:structure-recovery-markov}, we have that there exists $\hat{\tilde{X}}^n$ such that
$\hat{\tilde{X}}^n \in \Dc(S)$ and $\norm{\hat{\tilde{X}}^n - \hat{X}^n} \leq C 2^{-b}$, and hence we have:
\begin{equation*}
    \norm{Y^n - \hat{\tilde{X}}^n}_2^2 \leq \norm{Y^n - \tilde{X}^n} + 2 \lambda C 2^{-b} + n C / b.
\end{equation*}
However, noting that with probability 1, $\tilde{X}$ is a solution of a locally strongly convex
optimization problem, and that $\hat{\tilde{X}}^n$ is feasible, we deduce that:
\begin{equation*}
    \norm{\tilde{X}^n - \hat{\tilde{X}}^n}_2 \leq o(b),
\end{equation*}
from which we deduce that $\hat{X}^n \rightarrow \tilde{X}^n$, as $b \rightarrow \infty$ on
the event $\Ec_1 \cap \Ec_2$.

To conclude, we are also required to control $\norm{X^n - \hat{X}^n}$ on $\Ec_1^c$ and $\Ec_2^c$.
Note that $\PP(\Ec_2^c) \leq 2 \exp\{-\log^2(1 / \sigma)  / 2\} \leq 2\sigma^{-\log \sigma / 2}$ by
standard Gaussian concentration,
so that $\frac{1}{\sigma^2}\PP(\Ec_2^c) \rightarrow 0$ as $\sigma \rightarrow 0$. Hence on this
event we simply control $\norm{X - \hat{X}}$ by a constant.

On the other hand, note that $L(Y) \leq n \lambda (2 +  C / b)$, and $L(\hat{X}) \leq L(Y)$, from
which we deduce that:
\begin{equation*}
    \norm{Y - \hat{X}}_2^2 \leq n \lambda (2 +  C / b) \leq 2 n \lambda.
\end{equation*}
Additionally, on $\Ec_2$, we have that $\norm{Y - X}_2 \leq \sigma \log(1 / \sigma)$,
and hence we have on this event that:
\begin{equation*}
    \norm{X - \hat{X}}_2 \leq \sigma \log(1 / \sigma) + \sqrt{2 n } \lambda.
\end{equation*}
On the other hand, we have that:
\begin{align*}
    \PP(\Ec_1^c)
    &\leq \sum_{i = 1}^n \PP(d(Y_i, Y_{i + 1}) \leq \delta \mid S_i = 0)\PP(S_i = 0) \\
    &\leq \sum_{i = 1}^n C \delta,
\end{align*}
as $(Y_i, Y_{i + 1})$ has a continuous distribution with density uniformly bounded above for any $\sigma$
when $S_i = 0$. Hence we deduce that:
\begin{equation*}
    \frac{1}{\sigma^2}\EE[\norm{X - \hat{X}}_2^2 \mid \Ec_1^c, \Ec_2] \PP(\Ec_1^c) \leq \frac{1}{\sigma^2}\PP(\Ec_1^c) (\sigma \log (1 / \sigma) + \delta)^2,
\end{equation*}
from which we have that $\frac{1}{\sigma^2}\EE[\norm{X - \hat{X}}_2^2 \mid \Ec_1^c, \Ec_2] \PP(\Ec_1^c) \rightarrow 0$
as $\sigma \rightarrow 0$.

\end{proof}

\section{Numerical experiments}\label{sec:numerical}

As an illustration of this general approach, we consider its application to image denoising.
Image denoising is a well-studied task, and state-of-the-art methods  achieve impressive performance,
with two main strategies being non-local methods \cite{dabov2007image,BuadesCollMorel2005} and those based on neural networks
\cite{Ulyanov2018Deep}. We propose an initial implementation inspired by the Q-MAP estimator for denoising images.

\subsection{Implementation}
The  Q-MAP denoiser is a Bayesian method that  requires access to source distribution $\pi$. In image denoising, $\pi$ represents the distribution of natural images (or patches thereof), which is very challenging to estimate.  However, as described earlier, not only does the new framework lead us to a learning-based denoiser, it also guides us to the key properties of the distribution  (as opposed to its full characterization) that we need to estimate.     Recall that the ideal weights required by Q-MAP \eqref{eq:weights} are derived from the distribution of the quantized $k$-blocks. Moreover, as highlighted in the examples, only a few of such blocks (those that have a non-vanishing  probability, as $b\to \infty$) matter. Therefore,  we basically  only need to identify (learn) such special blocks (patches) and estimate their probabilities.

Concretely, our initial implementation considers
patches of $4 \times 4$ pixels, and a quantizer $Q : \RR^{4 \times 4} \rightarrow \{ 1, \dotsc, N\}$, where $N$
denotes the number of chosen codewords. Our quantizer $Q$ is designed with the help of some prior knowledge on
the structure of images, and is based on quantizing coordinates of the discrete wavelet transform of $U$. Exact details
are provided in the next section.

We then consider a patch-based Q-MAP, where given a patch $Y^{4 \times 4}$, we compute an estimate $\hat{X}^{4 \times 4}$
such that:
\begin{equation}
    \label{eq:implementation-patch}
    \hat{X}^{4 \times 4} = \argmin_{U \in \RR^{4 \times 4}} \norm{Y^{4 \times 4} - U}_2^2 + \lambda \log \hat{\pi}(Q(U)),
\end{equation}
and the estimate for the full image is obtained by processing each patch separately and averaging.
Here, $\hat{\pi}$ denotes a discrete distribution on the quantized codewords, learned from data.
We learn $\hat{\pi}$ by sampling natural images from the ImageNet dataset \cite{deng2009imagenet},
and computing the empirical distribution $\hat{\pi}$ of the quantized patches generated from the images.

\subsection{Quantizer}

In our theoretical investigation, we have considered a simple binary quantization of the signal. Unfortunately, such an approach
suffers from the curse of dimensionality, whereby the amount of information required to describe a $k$-th order joint distribution
$\pi_{k, b}$ is of the order $2^{kb}$ and thus increases exponentially with both the quantization level and the order of the dependency
considered. This causes difficulties, both in estimating the joint distribution $\pi_{k, b}$ accurately, and also in optimizing efficiently
over the postulated objective.

Instead, we consider a quantization scheme which can be efficiently optimized over, and which can more efficiently capture
typical distribution of images. In an initial implementation, we consider a variable rate quantization of an orthogonal
wavelet transform
of the $4 \times 4$ patch. More precisely, consider the 2-D DWT transform $\tilde{U}^{4 \times 4}$ of the observed signal.
The quantization $Q(U)$ is given by concatenating 4-bit codes for coordinates of total frequency 1 ($\tilde{U}_{1, 0}$ and $\tilde{U}_{0, 1}$),
3-bit codes for coordinates of total frequency 2, 2-bit codes for coordinates of total frequency 3, and 1-bit codes for coordinates
of total frequency 4. All higher frequencies are
ignored. This gives a total code size of $N = 2^{28}$. The breakpoints in the codes are chosen according to the empirical distribution
observed in the training data.

Note that due to the specific structure of this regularizer, the constrained problem $\min_{U : Q(U) = k} \norm{Y - U}_2^2$
can be solved in closed form by coordinate-wise projection (as the subset $Q(U) = k$ describes a rotated rectangular region).
This enables us to compute the patch estimate $\hat{X}^{4 \times 4}$ efficiently.
We anticipate that in tasks that are well understood (such as image denoising), performance can be further improved by choosing
better quantization schemes.

\subsection{Results}
We have computed $\hat{\pi}$ as the empirical distribution of patches extracted from a subset of 12500 images
from the ImageNet dataset, with 128 patches sampled per image, for a total of 1.6 million samples. The estimator
$\hat{\pi}$ is smoothed by adding one pseudo-occurrence to each codeword. As can be expected from a
highly structured dataset, the distribution of $\hat{\pi}$ is strongly skewed, and seems to follow a power law
(see \cref{fig:rank-prob}).
\begin{figure}
    \centering
    \includegraphics[width=.9\columnwidth]{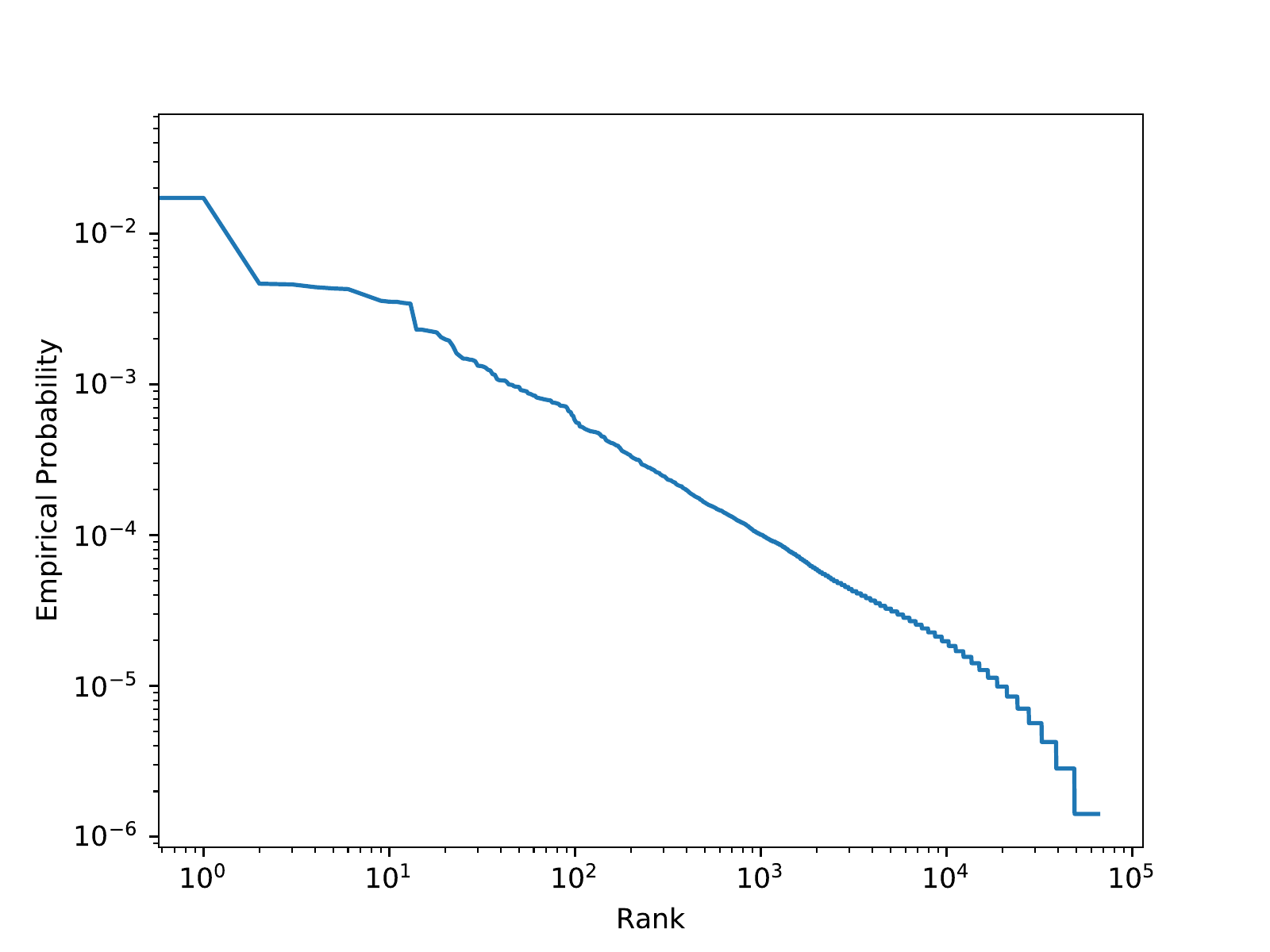}
    \caption{Rank-Probability Plot of $\hat{\pi}$\label{fig:rank-prob}}
\end{figure}

We use the learned distribution $\hat{\pi}$ to perform the denoising task on some standard image benchmarks
\cite{dabov2007image}, and present some preliminary results in \cref{table:comparison}. For reference,
we have also included a hard thresholding filter. It is interesting
to note that it corresponds to the Q-MAP filter with $\pi$ having a  uniform distribution over key features, further highlighting
the impact of the regularizer in the Q-MAP approach.

We note that the na\"ive patch estimator \eqref{eq:implementation-patch} can be
expensive, as it may require an exhaustive search over all possible codewords. However,  in practice,
most codewords have negligible probability, and can be ignored. Further speed-up may be possible
by pre-computing data structures to speed up the search.

\begin{table}
    \centering
    \begin{tabular}{rrrrrrr}
        \toprule
        & \multicolumn{4}{c}{PSNR} \\
        \cmidrule(lr){2-7}
        & \multicolumn{3}{c}{Camera} & \multicolumn{3}{c}{Peppers} \\
        \cmidrule(lr){2-4} \cmidrule(lr){5-7}
        $\sigma$ & Thresh & Q-MAP & BM3D & Thresh & Q-MAP & BM3D \\
        \midrule
        10 & 28.14 & 33.01 & 34.18 & 28.11 & 33.53 & 34.68 \\
        15 & 24.63 & 30.54 & 31.91 & 24.60 & 31.38 & 32.70 \\
        20 & 22.07 & 28.95 & 30.48 & 22.12 & 29.73 & 31.29 \\
        25 & 20.22 & 27.86 & 29.45 & 20.12 & 28.45 & 30.16 \\
        \bottomrule \\
    \end{tabular}

    \caption{Comparison between Q-MAP and BM3D \label{table:comparison}}
\end{table}

\section{Conclusions}\label{sec:conclusion}
In this paper we have studied the problem of denoising general analog stationary processes. In the Bayesian setting, where the source full distribution is known, we have proposed a new denoiser, Q-MAP denoiser. We have characterized the asymptotic    performance of the Q-MAP denoiser, as the power of noise approaches zeros, for i) stationary memoryless sources, and ii) structured 1-Markov sources. We have shown that the proposed method achieves optimal asymptotic  performance, at least for i.i.d.~sources. We have argued that the proposed method leads to a learning-based denoising algorithm. Initial results showing an application of the proposed learning-based method in image denoising is presented. 

\bibliographystyle{unsrt}
\bibliography{myrefs}

\begin{thebibliography}{10}

\bibitem{donoho1994ideal}
D.~L. Donoho and J.~M. Johnstone.
\newblock Ideal spatial adaptation by wavelet shrinkage.
\newblock {\em Biometrika}, 81(3):425--455, 1994.

\bibitem{stephane1999wavelet}
M.~Stephane.
\newblock A wavelet tour of signal processing.
\newblock {\em The Sparse Way}, 1999.

\bibitem{dabov2007image}
K.~Dabov, A.~Foi, V.~Katkovnik, and K.~Egiazarian.
\newblock Image denoising by sparse {3-D} transform-domain collaborative
  filtering.
\newblock {\em IEEE Trans. Image Processing}, 16(8):2080--2095, 2007.

\bibitem{burger2012image}
H.~C. Burger, C.~J. Schuler, and S.~Harmeling.
\newblock Image denoising: Can plain neural networks compete with {BM3D}?
\newblock In {\em Proc. of the IEEE Conf. on Comp. Vis. and Pat. Rec. (CVPR)},
  pages 2392--2399. IEEE, 2012.

\bibitem{Ulyanov2018Deep}
D.~Ulyanov, A.~Vedaldi, and V.~Lempitsky.
\newblock Deep image prior.
\newblock In {\em Proc. of the IEEE Conf. on Comp. Vis. and Pat. Rec. (CVPR)},
  pages 9446--9454, 2018.

\bibitem{jalali2018new}
S.~Jalali and A.~Maleki.
\newblock New approach to bayesian high-dimensional linear regression.
\newblock {\em Inf. and Inf.: A J. of the IMA}, 7(4):605--655, 2018.

\bibitem{deng2009imagenet}
J.~Deng, W.~Dong, R.~Socher, L.~J. Li, K.~Li, and L.~Fei-Fei.
\newblock Imagenet: A large-scale hierarchical image database.
\newblock In {\em Proc. of the IEEE Conf. on Comp. Vis. and Pat. Rec. (CVPR)},
  pages 248--255. Ieee, 2009.

\bibitem{cover}
T.~Cover and J.~Thomas.
\newblock {\em Elements of Information Theory}.
\newblock Wiley, New York, 2nd edition, 2006.

\bibitem{renyi1959dimension}
Alfr{\'e}d R{\'e}nyi.
\newblock On the dimension and entropy of probability distributions.
\newblock {\em Acta Math. Acad. Scien. Hungarica}, 10(1-2):193--215, 1959.

\bibitem{JalaliP:17-IT}
S.~Jalali and H.~V. Poor.
\newblock Universal compressed sensing for almost lossless recovery.
\newblock {\em IEEE Trans. Inform. Theory}, 63(5):2933--2953, May 2017.

\bibitem{wu2011mmse}
Y.~Wu and S.~Verd{\'u}.
\newblock {MMSE} dimension.
\newblock {\em IEEE Trans. Inform. Theory}, 57(8):4857--4879, 2011.

\bibitem{BuadesCollMorel2005}
A.~Buades, B.~Coll, and J.M. Morel.
\newblock A non-local algorithm for image denoising.
\newblock In {\em Proc. of the IEEE Conf. on Comp. Vis. and Pat. Rec. (CVPR)},
  volume~2, pages 60--65. IEEE, 2005.

\end{thebibliography}

\end{document}